\documentclass[letterpaper, 10 pt, conference]{ieeeconf}  

\IEEEoverridecommandlockouts                              
\usepackage{amsmath,amssymb,amsfonts}
\usepackage{graphicx}
\usepackage{textcomp}

\usepackage{graphicx}
\usepackage{epsfig} 
\usepackage{cite}

\usepackage{mathtools}
\usepackage{tikz}
\usepackage{fancyhdr}

\let\labelindent\relax
\usepackage{enumitem}

\fancypagestyle{ieeecopyright}{
  \fancyhf{}

  \fancyfoot[C]{
    \parbox{\textwidth}{\scriptsize
      ©2026 IEEE.  Personal use of this material is permitted.  Permission from IEEE must be obtained for all other uses, in any current or future media, including reprinting/republishing this material for advertising or promotional purposes, creating new collective works, for resale or redistribution to servers or lists, or reuse of any copyrighted component of this work in other works.%
    }%
  }
}

\newtheorem{theorem}{Theorem}

\newtheorem{proposition}{Proposition}

\newtheorem{definition}{Definition}
\newtheorem{remark}{Remark}
\newtheorem{assumption}{Assumption}

\newcommand{\bbR}{\mathbb{R}}
\newcommand{\bbN}{\mathbb{N}}

\newcommand{\bbC}{\mathbb{C}}

\newcommand{\calD}{\mathcal{D}}
\newcommand{\calA}{\mathcal{A}}

\newcommand{\calR}{\mathcal{R}}

\newcommand{\calL}{\mathcal{L}}

\newcommand{\calN}{\mathcal{N}}
\newcommand{\calG}{\mathcal{G}}
\newcommand{\calE}{\mathcal{E}}
\newcommand{\calM}{\mathcal{M}}
\newcommand{\calW}{\mathcal{W}}
\newcommand{\calP}{\mathcal{P}}
\newcommand{\calV}{\mathcal{V}}

\newcommand{\sfx}{\mathsf{x}}
\newcommand{\sfz}{\mathsf{z}}
\newcommand{\sfo}{\mathsf{o}}
\newcommand{\sfA}{\mathsf{A}}
\newcommand{\sfM}{\mathsf{M}}

\newcommand{\unobs}{{\bar{\mathcal O}}}

\DeclareMathOperator*{\diag}{diag}
\DeclareMathOperator*{\spn}{span}
\DeclareMathOperator*{\col}{col}
\DeclareMathOperator*{\im}{im}

\title{\LARGE \bf
On  Stability of a Distributed Observer Over Switching Networks:\\
The Case of Commutative Coupling
}

\author{Sunghyun Koo, Jin Gyu Lee, and Hyungbo Shim
\thanks{
This work was supported by the National Research Foundation of Korea (NRF) grant funded by the Korea government (MSIT) (RS-2026-25504174).\newline\indent
The authors are with ASRI, Department of Electrical and Computer Engineering, Seoul National University, Seoul 08826, South Korea (e-mail: {\tt\small shkoo@cdsl.kr}; {\tt\small jingyu.lee@snu.ac.kr}; {\tt\small hshim@snu.ac.kr}). Corresponding author: Jin Gyu Lee.}%
}

\begin{document}

\maketitle
\thispagestyle{ieeecopyright}
\pagestyle{empty}

\begin{abstract}

Distributed state estimation algorithms have been extended to switching networks, which model unreliable communication networks.
Yet existing results apply only to specific classes of systems, such as neutrally stable systems.
In this paper, we analyze a distributed observer under the assumption that the system matrix commutes with the coupling matrices.
This not only enables a refined analysis that leads to a necessary and sufficient condition for stability over switching networks, but also covers a broader class of systems than existing results, including unstable systems such as vehicular platoons.

\textit{Index Terms}---Distributed observer, distributed state estimation, multi-agent systems, switching networks.
\end{abstract}


\setlength\textfloatsep{4mm}
\setlength\floatsep{4mm}

\section{Introduction}

The distributed state estimation problem, which aims to estimate the state of a system using a multi-agent sensor network, has been extensively studied due to its wide range of applications \cite{survey}. This paper focuses on estimating the state of a continuous-time linear time-invariant system
\begin{equation}
\label{eqn:system}
    \begin{aligned}
    \dot x(t) &= Ax(t) \in\bbR^n\\
    y_i(t) &= C_i x(t) \in\bbR^{m_i}, \quad i\in\calN\coloneqq \{1,\dots,N\},
\end{aligned}
\end{equation}
where the output $y_i$ is the measurement of agent $i$, which may not suffice to reconstruct the state $x$ of the system (i.e., each pair $(C_i,A)$ may not be detectable).
In this case, to correctly estimate the full state $x$, the agents must communicate with other agents and appropriately incorporate the communicated information.
Various such algorithms, called \emph{distributed observers}, have been proposed \cite{park2017tac,mitra2018tac,kim2020TAC}.

In practice, the communication topology among agents may vary over time \cite{olfati2007ieee}.
One source of such topology changes is the unreliability of communication: communication links may intermittently fail because of packet losses or environmental effects, such as shadowing in wireless communication.
Another source arises when plug-and-play operation is allowed: as agents join and/or leave the network, communication links are established or removed.

Recent research efforts in distributed observer design have been successful in addressing these issues \cite{mitra2022TAC,xu2022tcyber,xu2025tac,zhao2025ecc,wang_split,wang2024auto,caiazzo2025ifac,yang2023auto,liu_continuous,zhang2024TAC,basu2024CDC}.
Yet, the problem remains open.
The algorithms proposed in \cite{mitra2022TAC}, \cite[Sec.~4]{xu2022tcyber}, \cite{xu2025tac} effectively operate over time-varying networks, but require a centralized design procedure and thus do not support plug-and-play operation.
The algorithms considered in \cite{zhao2025ecc,wang_split,wang2024auto,caiazzo2025ifac} allow both decentralized design and operation over switching networks, but require the network to relay all previously communicated information \cite{zhao2025ecc}, remain strongly connected despite switching \cite{wang_split}, or switch sufficiently fast\footnote{As discussed in \cite{xu2022tcyber}, fast switching may not be achievable in practice if the source of switching is unreliability of communication channels.} \cite{wang2024auto,caiazzo2025ifac}, thereby necessitating fairly reliable network conditions.
The algorithms introduced in \cite[Sec.~3]{xu2022tcyber}, \cite{yang2023auto,liu_continuous,zhang2024TAC,basu2024CDC} allow decentralized design while being well-suited for operation over unreliable networks.
However, they apply only to specific classes of systems: \cite[Sec.~3]{xu2022tcyber}, \cite{yang2023auto} require each pair $(A^\top,C_i^\top)$ to be marginally stabilizable, and \cite{liu_continuous,zhang2024TAC,basu2024CDC} require system \eqref{eqn:system} to be neutrally stable.\footnote{We call system \eqref{eqn:system} \emph{neutrally stable} if all the eigenvalues of $A$ lie on the imaginary axis and their algebraic and geometric multiplicities coincide.}

In line with \cite{liu_continuous,zhang2024TAC,basu2024CDC}, this paper also aims to extend the distributed observer proposed in \cite{kim2020TAC}, which supports decentralized design, to unreliable switching networks.
A key observation is that some practically relevant systems that are not neutrally stable, such as vehicular platoons \cite[Sec.~4.4]{wang2024auto} and battery packs \cite{khalil26auto}, still possess the ``commutative coupling structure" (see Section~\ref{sec:commutative_coupling_assumption}).
We analyze the distributed observer under this structure.
In turn, we establish a sharp necessary and sufficient condition for its stability over switching networks, from which several readily verifiable sufficient conditions follow.
The resulting conditions ensure exponential stability of the considered distributed observer over unreliable networks for a broader class of systems, encompassing the results of \cite{liu_continuous,zhang2024TAC,basu2024CDC} as special cases.

The rest of the paper is organized as follows. In Section~\ref{sec:preliminaries}, we present the graph-theoretic description of networks and define the notions of stability used in this paper.
In Section~\ref{sec:problem_setup}, we introduce the commutative coupling structure after a brief review of the distributed observer design proposed in \cite{kim2020TAC}.
The main theorems and their proofs are presented in Section~\ref{sec:stability_analysis}.
Section~\ref{sec:illustrative_example} illustrates the significance of the established conditions through an example, and the paper concludes in Section~\ref{sec:conclusion}.

\textit{Notation:}
For matrices or vectors $z_1,\dots,z_N$, we denote $\col(z_1,\dots,z_N) \coloneqq [z_1^\top,\dots,z_N^\top]^\top$ and denote by $\diag(z_1,\dots,z_N)$ the
block diagonal matrix with $z_1,\dots,z_N$ on the diagonal. The symbol $\otimes$ denotes the Kronecker product of matrices. The symbol $\oplus$ denotes the direct sum of subspaces.
The kernel and the image of a matrix are denoted by $\ker(\cdot)$ and $\im(\cdot)$, respectively.
The identity matrix of size $n\times n$ is denoted by $I_n$, and the zero matrix is denoted by $O$. For $p\in[1,\infty]$, $\|\cdot\|_p$ denotes the $p$-norm of a vector or the induced $p$-norm of a matrix.

\section{Preliminaries}
\label{sec:preliminaries}

\subsection{Description of Communication Network}

A communication topology can be represented as a directed graph.
Define the edge set $\calE\subseteq\calN\times \calN$ so that $(j,i)\in\calE$ if and only if agent $i$ can receive information from agent $j$.
Also, define the (weighted) adjacency matrix $\calA=[a^{ij}]\in\bbR^{N\times N}$ so that $a^{ij}$ represents the coupling strength of edge $(j,i)$.
Indeed, $a^{ij}\geq0$ and $a^{ij}>0\Leftrightarrow(j,i)\in\calE$.
A communication topology can then be characterized by the directed graph $\calG\coloneq(\calN,\calE,\calA)$.

To capture the dynamic aspect of communication, we assume that the communication network has a switching topology.
Let $\{\calG_p=(\calN,\calE_p,\calA_p):p\in\calP\}$ be a finite collection of all possible topologies of the communication network.
Then, a switching topology can be characterized by a switching directed graph $\calG_{\sigma(t)}$, where $\sigma:[0,\infty)\rightarrow\calP$ is a right-continuous piecewise constant switching signal.

\subsection{Stability Definitions}

We take the following stability definitions.
\begin{definition}
A dynamical system
\begin{equation}
\label{eqn:dynamical_system}
    \dot \sfz = f(t,\sfz) \in \bbR^{n_z}
\end{equation}
is called \emph{exponentially stable} if there exist a constant $r>0$ and a polynomial $p$ such that
\begin{equation}
\label{eqn:exponential_stability}
    \|\sfz(t)\|_2 \leq p(t-s) e^{-r(t-s)} \|\sfz(s)\|_2,\quad \forall t\geq s,
\end{equation}
for any initial condition $\sfz(s)\in\bbR^{n_z}$ and initial time $s$.
For $\calM\subseteq \bbR^{n_z}$ that is invariant under \eqref{eqn:dynamical_system}, system \eqref{eqn:dynamical_system} is said to be \emph{exponentially stable on $\calM$} if \eqref{eqn:exponential_stability} holds for any initial condition $\sfz(s)\in\calM$ and initial time $s$.
We refer to the constant $r$ as a \emph{guaranteed decay rate}.\footnote{We take no notice of polynomial growth when quantifying decay rates.}\hfill$\square$
\end{definition}

\section{Problem Setup}
\label{sec:problem_setup}

\subsection{Considered Distributed Observer Model}
\label{sec:considered_distributed_observer_model}

In this paper, we study the distributed observer proposed in \cite{kim2020TAC}.
Here, we briefly review its design.
Let $\hat x_i\in\bbR^n$ denote the state estimate maintained by agent $i$.
The algorithm for agent $i$ is given by
\begin{equation}
\label{eqn:distributed_observer}
    \dot{\hat x}_i = A\hat x_i - L_i (C_i \hat x_i - y_i) + \gamma_i M_i \sum_{j\in\calN} a^{ij}_{\sigma(t)}(\hat x_j - \hat x_i),
\end{equation}
where $L_i\in\bbR^{n\times m_i}$ is a local observer gain, $\gamma_i>0$ is a coupling gain, $M_i\in\bbR^{n\times n}$ is a coupling matrix, and $a^{ij}_{\sigma(t)}$ is the $ij$-th entry of the adjacency matrix $\calA_{\sigma(t)}$ 
corresponding to a switching network $\calG_{\sigma(t)}$.

The local observer gain $L_i$ and the coupling matrix $M_i$ can be synthesized in a decentralized manner as follows.
Let $\unobs_i$ denote the unobservable subspace of the pair $(C_i,A)$ and $n_i\coloneqq \dim \unobs_i$. One can construct matrices $T_{iu}\in\bbR^{n\times n_i}$ and $T_{io}\in\bbR^{n\times(n-n_i)}$ so that the columns of $T_{iu}$ and $T_{io}$ form orthonormal bases for $\unobs_i$ and $\unobs_i^\perp$, respectively.
This yields an observability decomposition
\begin{equation}
\label{eqn:observability_decomposition}
\begin{aligned}
    \begin{bmatrix}
        T_{io}^\top \\ T_{iu}^\top
    \end{bmatrix} A \begin{bmatrix}
        T_{io} & \kern-3mm T_{iu}
    \end{bmatrix} = \begin{bmatrix}
        A_{io} & \kern-3mm O \\ A_{ir} & \kern-3mm A_{iu}
    \end{bmatrix},~ C_i\begin{bmatrix}
        T_{io} & \kern-3mm T_{iu}
    \end{bmatrix} = \begin{bmatrix}
        C_{io} & \kern-3mm O
    \end{bmatrix},
\end{aligned}
\end{equation}
where the pair $(C_{io},A_{io})$ is observable.
The local observer gain $L_i$ is then defined as $L_i \coloneqq T_{io} L_{io}$, where $L_{io}\in\bbR^{(n-n_i)\times m_i}$ is chosen so that $A_{io}-L_{io} C_{io}$ is Hurwitz.
The coupling matrix is chosen as the orthogonal projection matrix onto $\unobs_i$, namely, $M_i \coloneqq T_{iu} T_{iu}^\top$.

\subsection{The Objective and Preliminary Observations}

Defining the error variable as $\tilde x_i \coloneqq \hat x_i - x$, one can verify that the error dynamics is given by
\begin{equation}
\label{eqn:error_dynamics}
    \dot{\tilde x}_i = (A-L_iC_i)\tilde x_i + \gamma_i M_i \sum_{j\in\calN} a^{ij}_{\sigma(t)} (\tilde x_j -\tilde x_i), ~ \forall i \in \calN.
\end{equation}
By stacking the error variables as $\tilde x \coloneqq \col(\tilde x_1,\dots,\tilde x_N)$, it can be compactly written as
\begin{equation}
\label{eqn:error_dynamics_stacked}
    \dot{\tilde x} = ( I_N \otimes A - L C) \tilde x - \Gamma M (\calL_{\sigma(t)}\otimes I_n) \tilde x,
\end{equation}
where $L\coloneqq \diag(L_1,\dots, L_N)$, $C\coloneqq \diag(C_1,\dots,C_N)$, $\Gamma\coloneqq \diag(\gamma_1I_n,\dots,\gamma_NI_n)$, $M\coloneqq\diag(M_1,\dots,M_N)$, and $\calL_{\sigma(t)}$ is the graph Laplacian of $\calG_{\sigma(t)}$.\footnote{The graph Laplacian of a directed graph $\calG$ is defined by $\calL\coloneqq \calD-\calA$, where $\calD$ is the diagonal matrix that makes each row sum of $\calL$ zero.}
The objective of this paper is to establish conditions that guarantee the stability of \eqref{eqn:error_dynamics_stacked} so that the distributed observer \eqref{eqn:distributed_observer} can be applied to a broader class of systems over unreliable networks.

We now present some preliminary observations on the error dynamics, which have appeared in the literature in various equivalent forms.
The first observation is the convergence of $(I-M_i) \tilde x_i$, which can be interpreted as the estimation error of the ``observable part."
Indeed, by pre-multiplying $T_{io}^\top$ to \eqref{eqn:error_dynamics}, we obtain $T_{io}^\top \dot{\tilde x}_i = (A_{io} - L_{io} C_{io}) T_{io}^\top \tilde x_i$. (Note from \eqref{eqn:observability_decomposition} that $T_{io}^\top A = A_{io} T_{io}^\top$ and $C_i = C_{io} T_{io}^\top$.)
Thus, we have
\begin{equation}
\label{eqn:observable_error}
    T_{io}^\top \tilde x_i(t) = e^{(A_{io} - L_{io} C_{io})(t-s)} T_{io}^\top \tilde x_i(s)
\end{equation}
for any $t\geq s$ and any initial condition $\tilde x_i(s)\in\bbR^n$.
Due to the design of $L_{io}$, we can conclude that $T_{io}^\top \tilde x_i$ converges exponentially to zero, and so does $(I-M_i)\tilde x_i=T_{io}T_{io}^\top \tilde x_i$.

In view of \eqref{eqn:error_dynamics_stacked}, this observation implies that $\tilde x$ exponentially gets close to the subspace
\begin{equation*}
    \calW \coloneqq \unobs_1\times \cdots \times \unobs_N.
\end{equation*}
Thus, it may suffice to analyze the stability of the error dynamics only on $\calW$.
Proposition~\ref{prop:preliminary_observation} justifies this intuition.

\begin{proposition}
\label{prop:preliminary_observation}
    The following statements hold:
    \begin{enumerate}
        \item The subspace $\calW$ is $(I_N\otimes A-LC)$- and $-\Gamma M(\calL_p\otimes I_n)$-invariant for each $p\in\calP$, and thus invariant under \eqref{eqn:error_dynamics_stacked}.
        \item The error dynamics \eqref{eqn:error_dynamics_stacked} is exponentially stable (on $\bbR^{nN}$) if and only if it is exponentially stable on $\calW$.\hfill$\square$ 
    \end{enumerate}
\end{proposition}
\textit{Proof:} 1) Since each $\unobs_i$ is the largest $A$-invariant subspace contained in $\ker(C_i)$, $\calW$ is clearly $(I_N\otimes A - LC)$-invariant.
Since $M_i$ is the orthogonal projection matrix onto $\unobs_i$, we also have that $\calW$ is $-\Gamma M (\calL_{\sigma(t)}\otimes I_n)$-invariant.

2) ($\Leftarrow$) Let $T_u \coloneqq \diag(T_{1u},\dots,T_{Nu})$ and $T_o\coloneqq \diag(T_{1o},\dots,T_{No})$. Since $\calW$ is invariant under \eqref{eqn:error_dynamics_stacked}, defining $\xi_u\coloneqq T_u^\top \tilde x$ and $\xi_o \coloneqq T_o^\top \tilde x$, the error dynamics \eqref{eqn:error_dynamics_stacked} can be represented as a block lower-triangular system
\begin{equation}
\label{eqn:lower_block_triag}
    \begin{bmatrix}
        \dot \xi_o\\\dot \xi_u
    \end{bmatrix} =
    \begin{bmatrix}
        *_{11} & O\\
        *_{21} & *_{22}
    \end{bmatrix} \begin{bmatrix}
        \xi_o \\ \xi_u
    \end{bmatrix}.
\end{equation}
By \eqref{eqn:observable_error}, the $\xi_o$-subsystem is exponentially stable.
Exponential stability of the $\xi_u$-subsystem follows from exponential stability of \eqref{eqn:error_dynamics_stacked} on $\calW$.
Since $\calP$ is finite, $*_{21}$ is bounded.
Thus, \cite[Thm.~2]{zhou2016auto} establishes exponential stability of \eqref{eqn:lower_block_triag}.

($\Rightarrow$) This direction is obvious.\hfill$\blacksquare$

In fact, the error dynamics \eqref{eqn:error_dynamics_stacked} can be further simplified on the invariant subspace $\calW$. Indeed, since $C_i\tilde x_i=0$ whenever $\tilde x_i\in\unobs_i$, the error dynamics \eqref{eqn:error_dynamics_stacked} reduces to 
\begin{equation}
\label{eqn:error_stacked_on_W}
    \dot {\tilde x} = (I_N \otimes A) \tilde x - \Gamma M (\calL_{\sigma(t)} \otimes I_n) \tilde x
\end{equation}
on $\calW$. (Obviously, $\calW$ is invariant under \eqref{eqn:error_stacked_on_W}.)
Thus, the stability analysis of the distributed observer~\eqref{eqn:distributed_observer} can be reduced to that of \eqref{eqn:error_stacked_on_W} on the invariant subspace $\calW$.

\subsection{Commutative Coupling Structure}
\label{sec:commutative_coupling_assumption}

We proceed with the analysis of \eqref{eqn:error_stacked_on_W} assuming the commutative coupling structure, which is first introduced in \cite{koo2025lcss} to develop a distributed state estimation algorithm that utilizes intermittent communication protocols.

\begin{assumption}
\label{ass:commuting_coupling}
    For each $i\in\calN$, $AM_i = M_iA$.\hfill$\square$
\end{assumption}

There are some equivalent conditions that are useful when verifying Assumption~\ref{ass:commuting_coupling}.

\begin{proposition}
\label{prop:equivalent_conditions}
    For $A$, $A_{ir}$, $\unobs_i$, and $M_i$ defined in Section~\ref{sec:considered_distributed_observer_model}, the following conditions are equivalent:%
    \vspace{2mm}

    \noindent
    \begin{minipage}{0.5\columnwidth}
        \begin{enumerate}%
            \item[(i)] $AM_i = M_iA$,
            \item[(ii)] $\unobs_i^\perp$ is $A$-invariant,
        \end{enumerate}
    \end{minipage}\begin{minipage}{0.5\columnwidth}
        \begin{enumerate}
            \item[(iii)] $\unobs_i$ is $A^\top$-invariant,
            \item[(iv)] $A_{ir}=O$.\hfill$\square$
        \end{enumerate}
    \end{minipage}
\end{proposition}
\vspace{2mm}
\textit{Proof:} In order for the projection matrix $M_i$ to commute with $A$, a necessary and sufficient condition is that both $\im(M_i)=\unobs_i$ and $\ker(M_i)=\unobs_i^\perp$ are $A$-invariant \cite[Thm.~6.10]{hoffman}. Since the unobservable subspace $\unobs_i$ is already $A$-invariant, conditions (i) and (ii) are equivalent. The equivalence of (ii) and (iii) can be shown as in \cite[Prop.~3.1.3]{invariant_subspace}. The equivalence of (ii) and (iv) can be shown based on the fact that $A_{ir} = T_{iu}^\top A T_{io}$.\hfill$\blacksquare$

The adoption of this ``commutative coupling structure" simplifies the analysis of \eqref{eqn:error_stacked_on_W} to a great extent and, in turn, facilitates a derivation of tight conditions for the stability of the distributed observer \eqref{eqn:distributed_observer}.
In the meantime, this assumption holds for a fairly broad class of systems.
In general, when the system matrix $A$ is diagonalizable (over $\bbC$), Assumption~\ref{ass:commuting_coupling} can always be satisfied up to a coordinate transformation \cite[Lemma~3]{koo2025lcss}.
This class of systems encompasses all neutrally stable systems considered in \cite{liu_continuous,zhang2024TAC,basu2024CDC}, as well as other practically relevant systems, such as the lithium-ion battery pack model considered in \cite{khalil26auto}.
Nevertheless, these are not the only systems that possess the commutative coupling structure.
For example, Assumption~\ref{ass:commuting_coupling} is satisfied in the vehicular platoon system considered in  \cite[Sec.~4.4]{wang2024auto}, although its system matrix is not diagonalizable.

\begin{remark}
    The validity of Assumption~\ref{ass:commuting_coupling} depends on the choice of coordinates,\footnote{The $A$-invariance of $\unobs_i^\perp$ is coordinate-dependent, unlike that of $\unobs_i$.} and hence an appropriate coordinate transformation may be necessary for Assumption~\ref{ass:commuting_coupling} to hold.
    Notably, when $A$ is diagonalizable, such a coordinate transformation can be performed without imposing any additional restrictions.
    Indeed, one always-applicable option for ensuring Assumption~\ref{ass:commuting_coupling} is to transform $A$ into its real Jordan form \cite[Lemma~3]{koo2025lcss}, which can be performed locally by each agent without requiring any global information.
    In fact, such coordinate transformations have been employed in \cite{liu_continuous,zhang2024TAC,basu2024CDC} for neutrally stable cases.\hfill$\square$
\end{remark}

\section{Stability Analysis}
\label{sec:stability_analysis}

\subsection{Necessary and Sufficient Condition for Stability}
\label{sec:necesary_and_sufficient_conditions}

Under Assumption~\ref{ass:commuting_coupling}, it can be seen that
\begin{equation*}
    (I_N\otimes A) \Gamma M(\calL_{\sigma(t)}\otimes I_n) = \Gamma M(\calL_{\sigma(t)}\otimes I_n) (I_N \otimes A)
\end{equation*}
regardless of $\calL_{\sigma(t)}$.
This implies that the vector fields $(I_N\otimes A)\tilde x$ and $-\Gamma M(\calL_{\sigma(t)}\otimes I_n)\tilde x$, which compose the dynamics \eqref{eqn:error_stacked_on_W}, are commuting vector fields.
Thus, the effects of two vector fields can be perfectly separated \cite[Thm.~5.1]{hall}.
To be precise, let $\Phi_{A+M}(t,s)$ denote the state transition matrix of \eqref{eqn:error_stacked_on_W}, $\Phi_A(t,s)\coloneqq e^{(I_N\otimes A)(t-s)}$, and $\Phi_M(t,s)$ denote the state transition matrix of the \emph{coupling-only system}
\begin{equation}
\label{eqn:coupling_only}
    \dot{\sfx}_i = \gamma_i M_i \sum_{j\in\calN} a_{\sigma(t)}^{ij}(\sfx_j-\sfx_i),\quad \forall i \in \calN,
\end{equation}
which can be compactly written as
\begin{equation}
\label{eqn:coupling_only_stacked}
    \dot{\sfx} = -\Gamma M(\calL_{\sigma(t)}\otimes I_n)\sfx
\end{equation}
for $\sfx\coloneqq\col(\sfx_1,\dots,\sfx_N)$.
Then, we have 
\begin{equation}
\label{eqn:flow_commute}
    \Phi_{A+M}(t,s) = \Phi_A(t,s)\Phi_M(t,s),
\end{equation}
for any $t,s\in[0,\infty)$.
A simple proof tailored to our linear time-varying case is provided in Appendix.

Through the lens of this separation, one can conclude that \eqref{eqn:error_stacked_on_W} becomes stable if the coupling-only system \eqref{eqn:coupling_only_stacked} decays fast enough to dominate the growth of $\Phi_A$.
However, the growth rate of $\Phi_A$ may differ across modes.
A more careful analysis is thus required for a tighter stability criterion.

Toward this end, we consider the primary decomposition for $A$ \cite[Sec.~6.8]{hoffman}, \cite[Sec.~12.2]{invariant_subspace}. Let $q_1^{\eta_1}(s) \cdots q_D^{\eta_D}(s)$ be the minimal polynomial of $A$, where $q_1,\dots,q_D$ are distinct irreducible polynomials over $\bbR$ and $\eta_1,\dots,\eta_D\in\bbN$. Then, we obtain the primary decomposition
\begin{equation}
\label{eqn:primary_decomposition}
    \bbR^n = \calR_1 \oplus \cdots \oplus \calR_D
\end{equation}
where $\calR_d\coloneqq \ker(q_d^{\eta_d}(A))\subseteq \bbR^n$ for each $d=1,\dots,D$.
By definition, each subspace $\calR_d$ is associated with either a real eigenvalue of $A$ or a pair of complex conjugate eigenvalues of $A$. We denote the former by $\lambda_d$ and the latter by $\lambda_d\pm i\mu_d$.
In fact, in the former case, $\calR_d$ is the generalized eigenspace corresponding to $\lambda_d$. In the latter case, $\calR_d$ can be characterized as the set of all real-valued vectors contained in the direct sum of the generalized eigenspaces in $\bbC^n$ corresponding to $\lambda_d + i \mu_d$ and $\lambda_d - i\mu_d$.
It can be shown that each $\calR_d$ is $A$-invariant.

For each $d=1,\dots,D$, we now define a subspace
\begin{equation*}
    \calW_d \coloneqq (\unobs_1 \cap \calR_d)\times \cdots \times (\unobs_N\cap \calR_d).
\end{equation*}
In fact, these subspaces form a direct sum decomposition of $\calW$ that allows the stability analysis of \eqref{eqn:error_stacked_on_W} on $\calW$ to be carried out separately on each $\calW_d$.

\begin{proposition}
\label{prop:w_d}
    The following statements hold:
    \begin{enumerate}
        \item $\calW = \calW_1\oplus \cdots \oplus \calW_D$.
        \item Under Assumption~\ref{ass:commuting_coupling}, each subspace $\calW_d$ is $(I_N\otimes A)$- and $-\Gamma M(\calL_{p}\otimes I_n)$-invariant for each $p\in\calP$, and thus invariant under \eqref{eqn:error_stacked_on_W}.
        \item Under Assumption~\ref{ass:commuting_coupling}, system \eqref{eqn:error_stacked_on_W} is exponentially stable on $\calW$ if and only if it is exponentially stable on $\calW_d$ for each $d=1,\dots,D$.\hfill$\square$
    \end{enumerate}
\end{proposition}
\textit{Proof:}
The following plays a key role: If $\calV$ is an $A$-invariant subspace, then $\calV=(\calV\cap \calR_1)\oplus\cdots\oplus (\calV\cap\calR_D)$ \cite[p.~263]{hoffman}.

1) Since $\unobs_i$ is $A$-invariant, it can be decomposed as
\begin{equation}
\label{eqn:invariant_subspace_decomposition}
    \begin{aligned}
    \unobs_i &= (\unobs_i \cap \calR_1)\oplus \cdots \oplus (\unobs_i \cap \calR_D).\\
    \end{aligned}
\end{equation}
Thus, we have $\calW = \calW_1\oplus \cdots \oplus \calW_D$.

2) By statement~1) of Proposition~\ref{prop:preliminary_observation}, it suffices to show that $\calR_d\times \cdots \times \calR_d$ is $(I_N\otimes A)$- and $-\Gamma M(\calL_p\otimes I_n)$-invariant. The former is due to the $A$-invariance of $\calR_d$. To show the latter, it suffices to show that $\calR_d$ is $M_i$-invariant for each $i\in\calN$.
As $\unobs_i^\perp$ is $A$-invariant under Assumption~\ref{ass:commuting_coupling}, we have
\begin{equation}
\label{eqn:invariant_subspace_decomposition_perp}
    \unobs_i^\perp = (\unobs_i^\perp \cap \calR_1)\oplus \cdots \oplus (\unobs_i^\perp \cap \calR_D).
\end{equation}
Combining \eqref{eqn:invariant_subspace_decomposition} and \eqref{eqn:invariant_subspace_decomposition_perp} then yields $\bbR^n = \unobs_i \oplus \unobs_i^\perp= (\unobs_i\cap\calR_1) \oplus (\unobs_i^\perp \cap \calR_1)\oplus \cdots \oplus (\unobs_i \cap \calR_D) \oplus (\unobs_i^\perp \cap \calR_D)$, from which one can show that $\calR_d = (\unobs_i \cap \calR_d) \oplus (\unobs_i^\perp \cap \calR_d)$. Now, suppose that $v \in \calR_d$. Then, there exist $v_1\in(\unobs_i\cap\calR_d)$ and $v_2 \in (\unobs_i^\perp\cap\calR_d)$ such that $v=v_1+v_2$. Therefore, we have $M_i v = v_1\in\calR_d$, showing that $\calR_d$ is $M_i$-invariant.

3) This is a direct consequence of 1) and 2).\hfill$\blacksquare$

As the growth rate of $\Phi_A$ on $\calW_d$ can be exactly quantified by $\lambda_d$, using \eqref{eqn:flow_commute}, we can obtain a necessary and sufficient condition for the stability of \eqref{eqn:error_stacked_on_W} on $\calW_d$.

\begin{proposition}
\label{prop:condition_on_Wd}
    System \eqref{eqn:error_stacked_on_W} is exponentially stable on $\calW_d$ if and only if $\lambda_d<0$ or the coupling-only system \eqref{eqn:coupling_only_stacked} is exponentially stable on $\calW_d$ with a guaranteed decay rate larger than $\lambda_d$.\hfill$\square$
\end{proposition}
\textit{Proof:} ($\Leftarrow$) By \eqref{eqn:flow_commute}, there exists a polynomial $p_1$ such that
\begin{equation}
    \label{eqn:growth_bound}
    \begin{aligned}
        \|\tilde x(t)\|_2 &= \|\Phi_A(t,s)(\Phi_M(t,s)\tilde x(s))\|_2\\
        &\leq p_1(t-s) e^{\lambda_d (t-s)} \|\Phi_M(t,s)\tilde x(s)\|_2
    \end{aligned}
\end{equation}
for any $t\geq s$ and $\tilde x(s)\in\calW_d$.
For the inequality, we use the fact that $\Phi_M(t,s)\tilde x(s) \in \calW_d$ if $\tilde x(s)\in\calW_d$, which is due to statement~2) of Proposition~\ref{prop:w_d}.
Note that the growth rate of $\Phi_A$ can be quantified by $\lambda_d$ on $\calW_d\subseteq\calR_d\times \cdots\times \calR_d$.

First, suppose that $\lambda_d<0$.
In this case, it suffices to show that $\|\Phi_M(t,s)\|_2\leq c_1$ for some constant $c_1$. Indeed, it is shown in \cite[Sec.~4-A]{liu2016acc} that the coupling-only system \eqref{eqn:coupling_only_stacked} is non-expansive in a norm defined by $\|\sfx\|_{2,\infty}\coloneqq \|\!\col(\|\sfx_1\|_2,\dots,\|\sfx_N\|_2)\|_\infty$. The equivalence between the norms $\|\cdot\|_{2,\infty}$ and $\|\cdot\|_2$ then implies that there exists a constant $c_1$ such that $\|\Phi_M(t,s)\|_2\leq c_1$ for all $t\geq s$.

Now, suppose that $\lambda_d\geq0$ and \eqref{eqn:coupling_only_stacked} is exponentially stable with guaranteed decay rate $r>\lambda_d$.
Then, there exists a polynomial $p_2$ such that
\begin{equation*}
    \|\Phi_M(t,x)\tilde x(s)\|_2 \leq p_2(t-s) e^{-r(t-s)} \|\tilde x(s)\|_2
\end{equation*}
for any $t\geq s$ and $\tilde x(s)\in\calW_d$. 
Combining this with \eqref{eqn:growth_bound} shows the exponential stability of \eqref{eqn:error_stacked_on_W} on $\calW_d$.

($\Rightarrow$) 
Suppose that system \eqref{eqn:error_stacked_on_W} is exponentially stable on $\calW_d$ with some guaranteed decay rate $r>0$.
From \eqref{eqn:flow_commute}, we have $\Phi_M(t,s) = \Phi_A(s,t)\Phi_{A+M}(t,s)$.
Thus, there exist polynomials $p_3$ and $p_4$ such that
\begin{align*}
    \|\Phi_M(t,s)\tilde x(s)\|_2 &= \|\Phi_A(s,t)( \Phi_{A+M}(t,s)\tilde x(s))\|_2\\
    &\leq p_3(t-s)e^{-\lambda_d(t-s)}\|\Phi_{A+M}(t,s)\tilde x(s)\|_2\\
    &\leq p_4(t-s) e^{-\lambda_d(t-s)} e^{-r(t-s)} \| \tilde x(s)\|_2
\end{align*}
for any $t\geq s$ and $\tilde x(s)\in\calW_d$.
Here, in the first inequality, we use the fact that $\Phi_{A+M}(t,s)\tilde x(s)\in\calW_d$ if $\tilde x(s)\in\calW_d$, which is again due to statement~2) of Proposition~\ref{prop:w_d}.
Therefore, if $\lambda_d\geq0$, we can conclude that the coupling-only system \eqref{eqn:coupling_only_stacked} is exponentially stable on $\calW_d$ with guaranteed decay rate $\lambda_d+r$, which is larger than $\lambda_d$.\hfill$\blacksquare$

Finally, combining Propositions~\ref{prop:preliminary_observation}, \ref{prop:w_d}, and \ref{prop:condition_on_Wd}, we can conclude that the stability of the distributed observer \eqref{eqn:distributed_observer} can be determined by examining the stability of the coupling-only system, which is indeed an easier task.

\begin{theorem}
\label{thm:only_stability}
    Suppose that Assumption~\ref{ass:commuting_coupling} holds. Then, the error dynamics \eqref{eqn:error_dynamics_stacked} of the distributed observer \eqref{eqn:distributed_observer} is exponentially stable if and only if, for each $d\in\{1,\dots,D\}$ such that $\lambda_d\geq0$, the coupling-only system \eqref{eqn:coupling_only_stacked} is exponentially stable on $\calW_d$ with a guaranteed decay rate larger than $\lambda_d$.~$\square$
\end{theorem}

Through similar arguments, one can also derive a necessary and sufficient condition for guaranteeing a prescribed decay rate.
Indeed, a larger decay rate of the coupling-only system is required for all unstable modes and the stable modes whose decay is slower than desired.
The eigenvalues of $A_{io}-L_{io}C_{io}$ must also be assigned appropriately during the design of the local observer gain $L_i$.

\begin{theorem}
\label{thm:stability_criterion}
    Suppose that Assumption~\ref{ass:commuting_coupling} holds. Then, the error dynamics \eqref{eqn:error_dynamics_stacked} of the distributed observer \eqref{eqn:distributed_observer} is exponentially stable with a guaranteed decay rate $r>0$ if and only if the following two hypotheses hold: 
    \begin{enumerate}
        \item For each $i\in\calN$, all eigenvalues of $A_{io} - L_{io} C_{io}$ have real parts less than or equal to $-r$.
        \item For each $d\in\{1,\dots,D\}$ such that $\lambda_d > -r$, the coupling-only system \eqref{eqn:coupling_only_stacked} is exponentially stable on $\calW_d$ with guaranteed decay rate $r+\lambda_d$.\hfill$\square$
    \end{enumerate}
\end{theorem}

\begin{remark}
    Let $T_{iu,d}$ be a matrix whose columns form an orthonormal basis for $\unobs_i \cap \calR_d$. Then, using the parametrization $\sfx_{i,d}\coloneqq T_{iu,d}^\top \sfx_{i}$, the dynamics of the coupling-only system \eqref{eqn:coupling_only} on $\calW_d$ can be equivalently represented as
\begin{equation}
\label{eqn:non_restricted}
    \dot\sfx_{i,d} = \gamma_i T_{iu,d}^\top \sum_{j\in\calN} a^{ij}_{\sigma(t)} (T_{ju,d}\sfx_{j,d} - T_{iu,d} \sfx_{i,d}),\quad \forall i\in\calN.
\end{equation}
We note that dynamics of the form \eqref{eqn:coupling_only} commonly appear in the context of distributed linear equation solvers, e.g., \cite{liu2016acc}, while
those of the form \eqref{eqn:non_restricted} commonly appear in the context of distributed observers, e.g., \cite{kim2020TAC,wang_split,liu_continuous}.\hfill$\square$
\end{remark}

\subsection{Stability of the Coupling-Only System}
\label{sec:stability_of_the_coupling_only_system}

We now investigate when the stability of the coupling-only system \eqref{eqn:coupling_only_stacked} can be guaranteed.
We first define some notions regarding observability.

\begin{definition}
\label{def:joint_obsv}
    The eigenvalue $\lambda_d$ (or the eigenvalue pair $\lambda_d\pm i\mu_d$) of $A$ associated with $\calR_d$ is said to be \emph{jointly observable} if $\unobs_1\cap\cdots\cap\unobs_N\cap \calR_d=\{0\}$.
    System \eqref{eqn:system} is said to be \emph{jointly observable} (resp. \emph{jointly detectable}) if all eigenvalues (resp. all eigenvalues with nonnegative real part) are jointly observable.\hfill$\square$
\end{definition}

One can show that the joint observability/detectability of system \eqref{eqn:system} based on Definition~\ref{def:joint_obsv} is equivalent to the observability/detectability of the pair $(\col(C_1,\dots,C_N),A)$, respectively, which are more commonly adopted definitions.

In fact, joint observability of an eigenvalue (or a pair) is necessary for the stability of the coupling-only system on the corresponding invariant subspace.

\begin{proposition}
\label{prop:necessity_of_joint_observability}
    In order for the coupling-only system~\eqref{eqn:coupling_only_stacked} to be exponentially stable on $\calW_d$, it is necessary that $\lambda_d$ (or $\lambda_d\pm i\mu_d$) is jointly observable.\hfill$\square$
\end{proposition}
\textit{Proof:} Suppose that $\lambda_d$ (or $\lambda_d \pm i\mu_d$) is not jointly observable.
Set $\sfx_1=\cdots = \sfx_N \in (\unobs_1\cap\cdots \cap \unobs_N \cap \calR_d) \setminus \{0\}$. Then, it can be seen from \eqref{eqn:coupling_only} that all $\sfx_i\neq 0$ remain constant.\hfill$\blacksquare$

To guarantee the stability of the coupling-only system, we additionally need a condition regarding network connectivity.
One is the uniform joint strong connectedness defined below. We denote the sequence of switching instants by $0=t_1,t_2,\dots$, i.e., $\sigma(t)$ remains unchanged on each $[t_k,t_{k+1})$.

\begin{definition}
\label{def:joint_connectedness}
    A switching network represented by $\calG_{\sigma(t)}$ is \emph{uniformly jointly strongly connected} if there exists a dwell-time $\Delta>0$ such that $t_{k+1}-t_{k}\geq\Delta$ for each $k\in\bbN$ and there exist a constant $T>0$ and a subsequence $0=t_{k_1}, t_{k_2},\dots$ of switching instants such that, for each $l\in\bbN$, $t_{k_{l+1}}-t_{k_l} \leq T$ and the graph $\big(\calN,\bigcup_{k'=k_l}^{k_{l+1}-1} \calE_{\sigma(t_{k'})}\big)$ is strongly connected.~$\square$
\end{definition}

Thanks to \cite{liu2016acc}, exponential stability of \eqref{eqn:coupling_only_stacked} on $\calW_d$ can be guaranteed under the uniform joint strong connectedness.

\begin{proposition}
\label{prop:ujsc}
    If the eigenvalue $\lambda_d$ (or the eigenvalue pair $\lambda_d\pm i\mu_d$) of $A$ associated with $\calR_d$ is jointly observable, and the communication network is uniformly jointly strongly connected, then the coupling-only system \eqref{eqn:coupling_only_stacked} is exponentially stable on $\calW_d$.\hfill$\square$
\end{proposition}
\textit{Proof:} Let $\sfx\coloneqq \col(\sfx_1,\dots,\sfx_N)$ be the solution of \eqref{eqn:coupling_only_stacked} with initial condition $\sfx(s) \in \calW_d$.
On the one hand, it is shown in \cite[Sec.~IV-B]{liu2016acc} that all $\sfx_i$ asymptotically converge to a single point in $\unobs_1 \cap \cdots \cap \unobs_N$.\footnote{It can be seen that Definition~\ref{def:joint_connectedness} implies the repeated joint strong connectedness defined in \cite{liu2016acc}. The analysis in \cite[Sec.~Iv-B]{liu2016acc} remains valid in the presence of the edge weights $a^{ij}_{\sigma(t)}$ and the coupling gains $\gamma_i$.} On the other hand, since $\calR_d$ is $M_i$-invariant (see the proof of Proposition~\ref{prop:w_d}) and $\sfx_i(s)\in\calR_d$, each $\sfx_i$ remains in $\calR_d$. Therefore, we can conclude that each $\sfx_i$ asymptotically converges to a single point in $\unobs_1\cap \cdots \cap \unobs_N \cap \calR_d$, which must be $0$. 
Finally, given that system \eqref{eqn:coupling_only_stacked} is linear and the convergence is uniform, we can conclude that \eqref{eqn:coupling_only_stacked} is exponentially stable on $\calW_d$ \cite[Thm.~4.11]{khalil}.\hfill$\blacksquare$

In view of Theorems~\ref{thm:only_stability} and \ref{thm:stability_criterion}, it is sometimes necessary to guarantee a certain decay rate for the coupling-only system.
However, Proposition~\ref{prop:ujsc} guarantees only the stability and provides no information about the decay rate.
When the communication network is fixed, once stability is guaranteed, any prescribed decay rate can be achieved by increasing the coupling gains $\gamma_i$ over a certain threshold \cite{kim2020TAC}.
Unfortunately, without increasing the switching speed, this does not generally hold for uniformly jointly strongly connected networks.
We thus introduce a stronger connectivity condition.

\begin{definition}
\label{def:recurrent_connectedness}
    A switching network represented by $\calG_{\sigma(t)}$ is \emph{uniformly recurrently strongly connected} if there exist constants $T,\Delta>0$ and a subsequence $0=t_{k_1}, t_{k_2},\dots$ of switching instants such that, for each $l\in\bbN$, $t_{k_{l+1}}\!-t_{k_l} \leq T$, and that, for each $l\in\bbN$, there exists $k'\in\{k_{l},k_{l}+1,\dots,k_{l+1}-1\}$ for which the graph $\calG_{\sigma(t_{k'})}$ is strongly connected and $t_{k'+1}\!-t_{k'}\geq \Delta$.\hfill$\square$
\end{definition}

The assumption of uniform recurrent strong connectedness would be plausible in many practical settings where communication networks can be expected to become strongly connected at least for brief periods.
Under this assumption, a prescribed decay rate can be achieved by simply increasing the coupling gains over a threshold.

\begin{proposition}
\label{prop:ursc}
    If the eigenvalue $\lambda_d$ (or the eigenvalue pair $\lambda_d\pm i\mu_d$) of $A$ associated with $\calR_d$ is jointly observable, and the communication network is uniformly recurrently strongly connected, then for any $r>0$, there exists $\gamma^\ast>0$ such that the coupling-only system \eqref{eqn:coupling_only_stacked} is exponentially stable on $\calW_d$ with guaranteed decay rate~$r$ whenever $\gamma_i\geq \gamma^\ast$.\hfill$\square$
\end{proposition}
\textit{Proof:} Let $\calP^\sfo\subseteq\calP$ be the set of indices $p$ for which $\calG_p$ is strongly connected.
Following the proof of \cite[Thm.~1]{kim2020TAC} with $A$ set to $O$, one can show that \eqref{eqn:non_restricted} with topology fixed at $\calG_p$, $p\in\calP^\sfo$, can be made exponentially stable with an arbitrary guaranteed decay rate by increasing the coupling gains. In particular, for each $p\in\calP^\sfo$, there exists $\gamma_p^\ast$ such that
\begin{equation*}
    \|e^{-\Delta\Gamma M(\calL_p\otimes I_n)}\sfx\|_{2,\infty} \leq e^{-rT} \|\sfx\|_{2,\infty},\quad \forall \sfx\in\calW_d,
\end{equation*}
whenever $\gamma_i\geq \gamma_p^\ast$ for all $i\in\calN$. (See the proof of Proposition~\ref{prop:condition_on_Wd} for the definition of the norm  $\|\cdot\|_{2,\infty}$.) Let $\gamma^\ast \coloneqq \max_{p\in\calP^\sfo}\gamma_p^\ast$. Then, since $\|\sfx\|_{2,\infty}$ is non-expansive along the dynamics of the coupling-only system \eqref{eqn:coupling_only_stacked}, we can conclude from the uniform recurrent joint connectedness that
\begin{equation*}
    \|\sfx(t_{k_{l+1}})\|_{2,\infty}\leq e^{-rT}\|\sfx(t_{k_{l}})\|_{2,\infty}
\end{equation*}
for each $l\in\bbN$ whenever $\gamma_i\geq \gamma^\ast$ for all $i \in \calN$.
The proof is completed as $t_{k_{l+1}}-t_{k_l}\leq T$ and $\|\cdot\|_{2,\infty}$ is non-expansive.~$\blacksquare$

\subsection{Off-the-shelf Sufficient Conditions for Stability}
\label{sec:off_the_shelf_sufficient_conditions}

Using the conditions presented in Section~\ref{sec:stability_of_the_coupling_only_system}, we now establish three sufficient conditions for the stability of the distributed observer.
These conditions are practically verifiable as they are stated in terms of observability and connectivity.
We first restrict our attention to the systems with at most polynomial growth.

\begin{theorem}
\label{thm:sufficient_condition_1}
    Suppose that Assumption~\ref{ass:commuting_coupling} holds and system \eqref{eqn:system} is stable, marginally stable, or unstable with polynomial growth. If system~\eqref{eqn:system} is jointly detectable and the communication network is uniformly jointly strongly connected, then the error dynamics \eqref{eqn:error_dynamics_stacked} is exponentially stable.\hfill$\square$
\end{theorem}
\textit{Proof:} Combine Theorem~\ref{thm:only_stability} with Proposition~\ref{prop:ujsc}.\hfill$\blacksquare$

Theorem~\ref{thm:sufficient_condition_1} extends the results of \cite{liu_continuous,zhang2024TAC,basu2024CDC} in two perspectives.
First, it applies to a broader class of systems, including all neutrally stable systems. (Recall that Assumption~\ref{ass:commuting_coupling} is always valid when system~\eqref{eqn:system} is neutrally stable.)
Second, the convergence is guaranteed over directed networks; the previous works only considered undirected networks \cite{liu_continuous,zhang2024TAC}, or imposed an additional technical assumption (that $\unobs_i^\perp \perp \unobs_j^\perp$ if $i\neq j$, or equivalently, $\unobs_j^\perp \subseteq \unobs_i$) while considering directed networks \cite{basu2024CDC}.
When compared with \cite{wang2024auto,caiazzo2025ifac}, it requires neither fast switching nor large coupling gain, provided that the requirements are verified.

For exponentially unstable systems, a certain level of decay rate is required for the coupling-only system \eqref{eqn:coupling_only_stacked}. This can be achieved by increasing the coupling gains when the network is uniformly recurrently strongly connected.

\begin{theorem}
\label{thm:sufficient_condition_2}
    Suppose that Assumption~\ref{ass:commuting_coupling} holds.
    If system~\eqref{eqn:system} is jointly detectable and the communication network is uniformly recurrently strongly connected, then there exists $\gamma^\ast>0$ such that the error dynamics \eqref{eqn:error_dynamics_stacked} is exponentially stable whenever $\gamma_i>\gamma^\ast$ for each $i\in\calN$.\hfill$\square$
\end{theorem}
\textit{Proof:} Combine Theorem~\ref{thm:only_stability} with Proposition~\ref{prop:ursc}.\hfill$\blacksquare$

Regardless of the stability of system~\eqref{eqn:system}, to achieve an arbitrary decay rate, joint observability is required.

\begin{theorem}
\label{thm:sufficient_condition_3}
    Suppose that Assumption~\ref{ass:commuting_coupling} holds.
    If system~\eqref{eqn:system} is jointly observable and the communication network is uniformly recurrently strongly connected, then for any $r>0$, there exists $\gamma^\ast>0$ such that the error dynamics \eqref{eqn:error_dynamics_stacked} is exponentially stable with guaranteed decay rate $r$ whenever $\gamma_i>\gamma^\ast$ for each $i\in\calN$ and the eigenvalues of $A_{io}-L_{io}C_{io}$ have real parts less than or equal to $-r$ for each $i\in\calN$.\hfill$\square$
\end{theorem}
\textit{Proof:} Combine Theorem~\ref{thm:stability_criterion} with Proposition~\ref{prop:ursc}.\hfill$\blacksquare$

Compared with \cite[Thm.~2]{wang_split}, Theorems~\ref{thm:sufficient_condition_2} and \ref{thm:sufficient_condition_3} do not require the communication network to be strongly connected at all times.
To the best of the authors' knowledge, the stability of the distributed observer \eqref{eqn:distributed_observer} under uniform recurrent strong connectedness has not yet been studied.
We partially tackled this problem under Assumption~\ref{ass:commuting_coupling}.

\begin{remark}
    In practice, increasing the coupling gains $\gamma_i$ may amplify communication noise or lead to instability after digital discretization.
    In fact, for a fixed graph, the coupling gain required by Theorems~\ref{thm:sufficient_condition_1}--\ref{thm:sufficient_condition_3} is always no greater than that required by \cite[Thm.~1]{kim2020TAC}.
    The key difference is that the latter depends on $\|A_{iu}\|_2$, whereas the former depends on the maximum real part of the eigenvalues of $A_{iu}$, which is always no greater than $\|A_{iu}\|_2$.
    This suggests that applying a coordinate transformation that ensures Assumption~\ref{ass:commuting_coupling}, when possible, could help reduce coupling gains.\hfill$\square$
    
\end{remark}

\section{Illustrative Example}
\label{sec:illustrative_example}

In this section, we illustrate through an example how the established conditions guarantee convergence of the distributed observer \eqref{eqn:distributed_observer} in cases that have not been covered by previous results.
Motivated by \cite[Sec.~4.4]{wang2024auto}, consider a platoon of three vehicles whose dynamics is given by
\begin{equation*}
    \dot x = \underbrace{\left( I_3 \otimes \begin{bmatrix}
        0 & 1 \\ 0 & 0
    \end{bmatrix} \right)}_{=A} x,\vspace{-2mm}
\end{equation*}
where $x \coloneqq \col(p_1,v_1,p_2,v_2,p_3,v_3)\in\bbR^6$ with $p_i$ and $v_i$ being the position and velocity of vehicle $i$, respectively.
We assume that vehicle~1 can measure its absolute position (i.e., $y_1=p_1$), while vehicles~2 and~3 can only measure their relative positions from the preceding vehicle (i.e., $y_2=p_2-p_1$ and $y_3=p_3-p_2$).
The goal of each vehicle is to estimate the state of the whole platoon using its measurement while communicating with other vehicles.
The considered switching communication network is shown in Fig.~\ref{fig:graph}.

\setlength\abovecaptionskip{0mm}

\begin{figure}
    \centering
    \vspace{1.3mm}
    \includegraphics[width=1\linewidth]{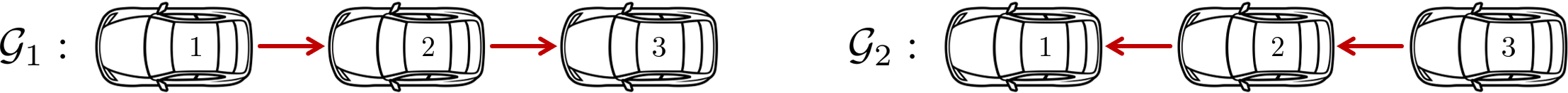}
    \caption{The considered communication network. Two digraphs switch every 0.5 seconds. All edge weights are 1.}
    \label{fig:graph}
\end{figure}

\setlength\abovecaptionskip{0mm}

\begin{figure}
    \centering
    \includegraphics[width=1\linewidth]{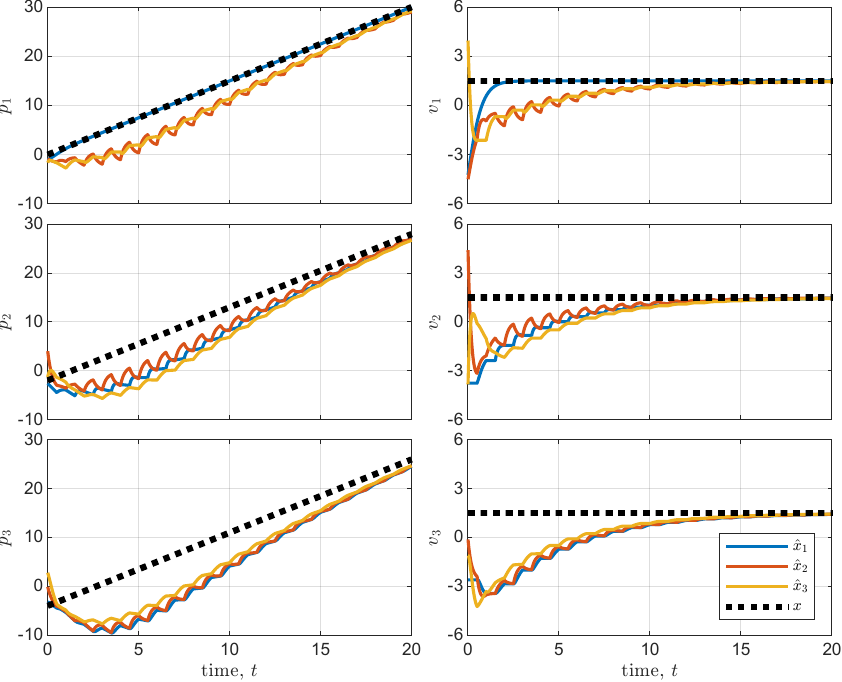}
    \caption{Results of a numerical simulation with $\gamma_i=5$. The black dotted lines depict the true state, and the colored solid lines depict the state estimate made by each vehicle.}
    \label{fig:sim_figure}
\end{figure}

Because the system is unstable and the communication network is not always strongly connected, it falls outside the scope of the existing results.
However, it can be seen that all the requirements of Theorem~\ref{thm:sufficient_condition_1} are satisfied.
To verify Assumption~\ref{ass:commuting_coupling}, let $\{e_1,\dots,e_6\}$ denote the standard basis for $\bbR^6$. Then, one can check that
$\unobs_1^\perp = \spn\{e_1,e_2\}$, $\unobs_2^\perp = \spn\{e_3-e_1,e_4-e_2\}$, and $\unobs_3^\perp = \spn\{e_5-e_3,e_6-e_4\}$ are all $A$-invariant. (In fact, Assumption~\ref{ass:commuting_coupling} is always valid whenever each $y_i$ is a function of $p_1$, $p_2$, and $p_3$.)
Therefore, the distributed observer \eqref{eqn:distributed_observer} is exponentially stable for any choice of coupling gains $\gamma_i$.
This conclusion is also supported by a numerical simulation.
It can be seen in Fig.~\ref{fig:sim_figure} that the state estimate of each vehicle converges to the true state.

We also demonstrate through a numerical simulation that uniform joint strong connectedness is insufficient for achieving an arbitrary decay rate.
Fig.~\ref{fig:sim_figure_high_gain} shows the results of a numerical simulation where we set $\gamma_i=500$ and the eigenvalues of $A_{io}-L_{io}C_{io}$ to have real parts less than $-50$.
Here, the observed decay rate is not particularly fast compared to that in Fig.~\ref{fig:sim_figure}.
To achieve an arbitrary decay rate, a stronger connectivity condition, such as uniform recurrent strong connectedness, is needed.

\begin{figure}
    \centering
    \vspace{1.2mm}
    \includegraphics[width=1\linewidth]{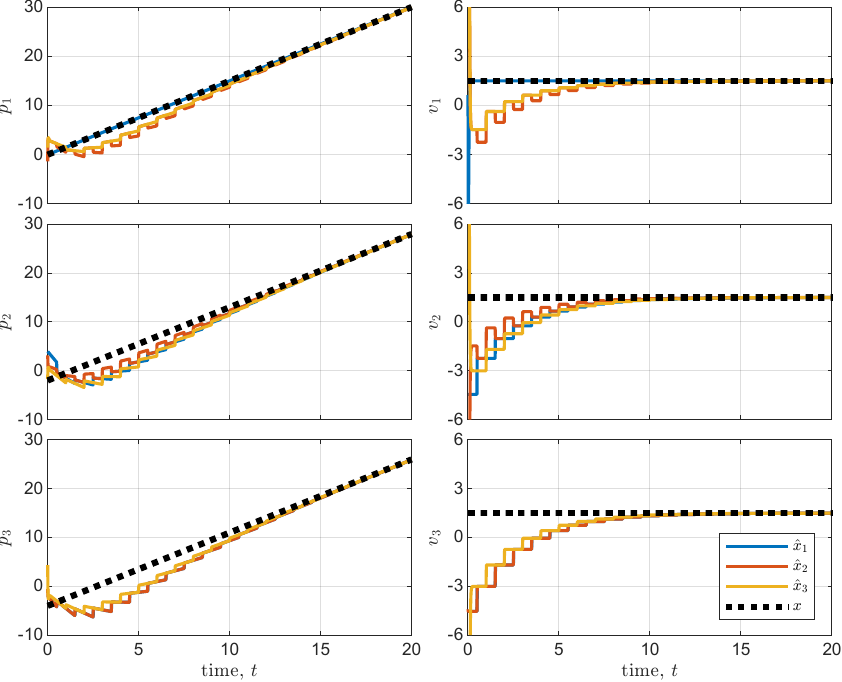}
    \caption{Results of a numerical simulation with $\gamma_i=500$. The lines represent the same quantities as in Fig.~\ref{fig:sim_figure}.}
    \label{fig:sim_figure_high_gain}
\end{figure}

\begin{remark}
\label{rmk:discarding}
    Suppose now that vehicle 2 can additionally measure its velocity $v_2$ so that $y_2 = \col(p_2-p_1, v_2)$.
    Then, it can be seen that
    $\unobs_2^\perp = \spn\{e_3-e_1,e_2,e_4\}$,
    which is no longer $A$-invariant.
    Hence, by Proposition~\ref{prop:equivalent_conditions}, Assumption~\ref{ass:commuting_coupling} is violated, and in fact, the distributed observer \eqref{eqn:distributed_observer} becomes unstable over certain switching networks.
    Nevertheless, an exponentially stable distributed observer can still be constructed; one can simply discard the measurement $v_2$ and then construct a distributed observer.
    This remedy, however, comes at the cost of losing information and may therefore degrade estimation performance, especially in the presence of measurement noise.
    Developing a method that effectively exploits all available measurements while ensuring stability over switching networks is left as a future endeavor.\hfill$\square$
\end{remark}

\section{Conclusion}
\label{sec:conclusion}

Under the commutative coupling structure, $AM_i=M_iA$, we establish a necessary and sufficient condition for exponential stability of the distributed observer \eqref{eqn:distributed_observer}.
We also provide several sufficient conditions in terms of observability and connectivity, expanding the class of systems for which the distributed observer can be deployed over unreliable networks.
Extending this result to general linear systems and nonlinear systems can be a future direction for research.

\section*{Appendix}

Denote $\sfA\coloneqq I_N\otimes A$ and $\sfM(t) \coloneqq -\Gamma M(\calL_{\sigma(t)}\otimes I_n)$.
Let $z(t) \coloneqq \Phi_A(t,s)\Phi_M(t,s)z_0$ for an arbitrary $z_0\in\bbR^{nN}$.
Clearly, $z(s) = z_0$. 
Meanwhile, noting that $M(t)$ commutes with $\Phi_A(t,s)=e^{\sfA(t-s)}$, differentiating $z(t)$ yields
\begin{align*}
    \dot z(t) = \sfA \Phi_A(t,s)\Phi_M(t,s)z_0 + \Phi_A(t,s) \sfM(t) \Phi_M(t,s) z_0\\
    = (\sfA + \sfM(t)) \Phi_A(t,s)\Phi_M(t,s) z_0 = (\sfA+\sfM(t)) z(t).
\end{align*}
Hence, we have $z(t)=\Phi_{A+M}(t,s)z_0$.
Since $z_0$ is arbitrary, we can conclude that $\Phi_{A+M}(t,s)=\Phi_A(t,s)\Phi_M(t,s)$.

\bibliographystyle{IEEEtran}
\bibliography{Ref}
\end{document}